\documentclass[pra,preprint,aps]{revtex4}
\usepackage[T1]{fontenc}
\usepackage{amsthm}
\usepackage{amssymb}
\usepackage{graphicx}
\providecommand{\tabularnewline}{\\}
\theoremstyle{plain}
\newtheorem*{thm*}{Theorem}
\makeatother

\begin{document}

\newcommand{\be}{\begin{equation}}
\newcommand{\ee}{\end{equation}}
\newcommand{\ben}{\begin{eqnarray}}
\newcommand{\een}{\end{eqnarray}}
\newcommand{\ra}{\rangle}
\newcommand{\la}{\langle}
\newcommand{\ov}{\overline}
\newcommand{\kn}{| n \rangle}
\newcommand{\bn}{ \langle n |}
\newcommand{\til}{\tilde}
\newcommand{\iii}{\'{\i}}

\title{Quantum Potentials with q-Gaussian Ground States}

 \author{C. Vignat$^{1}$, A. Plastino$^{2}$,
 A.R. Plastino$^{2,3}$ and J.S. Dehesa$^{3}$}
 \address{$^{1}$L.S.S., Supelec, Université d'Orsay,
 Paris, France  \\
 $^{2}$National University La Plata, UNLP-CREG-CONICET,
 Casilla de Correos 727, 1900 La Plata, Argentina\\
 $^{3}$Instituto de F\'{\i}sica Te\'orica y Computacional
 Carlos I and  Departamento de F\'{\i}sica At\'omica, Molecular y Nuclear,
 Universidad de Granada, Granada, Spain, EU}
\begin{abstract}
 We determine families of spherically symmetrical $D$-dimensional
 quantum potential functions $V(r)$ having ground state wavefunctions
 that exhibit, either in configuration or in momentum space, the
 form of an isotropic $q$-Gaussian. These wavefunctions admit
 a maximum entropy description in terms of $S_q$ power-law
 entropies. We show that the potentials with a ground state
 of the $q$-Gaussian form in momentum space admit the Coulomb
 potential $-1/r$ as a particular instance. Furthermore, all
 these potentials behave asymptotically as the Coulomb
 potential for large $r$
 for all values of the parameter $q$ such that $0<q<1.$
\end{abstract}
\maketitle

\section{Introduction}

Extended versions of the maximum entropy principle based upon
power law $S_q$ entropies \cite{GT04,T09} have been found to
provide useful tools for the description of several physical
systems or processes \cite{GT04,T09,PP93a,RPCP00,ASLMM05,PA97}.
Indeed, various important equations in mathematical physics admit
exact solutions of the maximum $S_q$ form such as, for example,
the polytropic solutions to the Vlasov-Poisson equations
\cite{PP93a}, time dependent solutions to some evolution equations
involving non linear, power-law diffusion terms
\cite{RPCP00,ASLMM05}, or stationary phase-space distributions for
Liouville equations describing anomalous thermostatting processes
\cite{PA97}.

The application of information theoretical ideas to the study of
the eigenstates of diverse quantum systems has attracted the
attention of researchers in recent years
\cite{MYD10,LAD09,GDPS09,DLOY06,RN08,DB09,NR09,PCPP95}. The
standard maximum entropy principle, based on the optimization of
Shannon's entropic measure under appropriate constraints, plays a
distinguished role within these lines of enquiry. This principle
has been successfully applied to the characterization of the
eigenstates of various quantum systems (see, for instance,
\cite{NR09,PCPP95,PP93b} and references therein). In particular,
it is well known that the probability densities in both position
and momentum space corresponding to the ground state of the
isotropic $D$-dimensional quantum harmonic oscillator are
Gaussians, which are probability densities maximizing Shannon'
entropy under the constraints imposed by normalization and the
expectation value of the square $r^2$ of the radial coordinate.

It would be of considerable interest to extend to the $S_q$-based
framework the maximum entropy approach to the description of the
eigenstates of quantum systems. This formalism has already been
applied to the study of various quantum phenomena (see, for
example \cite{BCPP02,TBD98}), but its application to characterize
the probability densities associated with quantum eigenstates
remains largely unexplored. The maximum entropy formalism based on
the $S_q$ entropies leads to a generalization of the Gaussian
probability density, which is given by the so-called $q$-Gaussians
\cite{GT04,T09}. These $q$-Gaussian constitute some of the
simplest and most important examples of maximum-$S_q$
distributions. The aim of the present work is to determine the
form of those spherically symmetric quantum potentials $V(r)$
whose ground state wavefunctions (in position or in momentum
space) are associated with $q$-Gaussian densities.

\section{q-Gaussian Ground States in Configuration Space}

We are going to consider a spinless particle of mass $m$ in a
$D$-dimensional configuration space. The eigenfunctions $\psi({\bf
r })$ associated with a potential $V({\bf r})$ obey then the
Schrödinger equation,

\begin{equation}
-\frac{\hbar^{2}}{2m}\nabla^{2}\psi+V\psi=E\psi,\label{eq:Schrodinger}
\end{equation}

\noindent where $\nabla^{2}$ is the $D$-dimensional Laplacian
operator, $\hbar$ is Planck's constant and $E$ is the energy
eigenvalue. We assume in the rest of this paper that $m=\hbar=1$.
Since we are going to consider spherically symmetric potentials,
the Schrödinger equation for the concomitant ground states (which
are spherically symmetric) simplifies to

\be -\frac{1}{2 r^{D-1}}\frac{\partial}{\partial
r}\left(r^{D-1}\frac{\partial\psi}{\partial r}\right)+V\psi=E\psi,
\ee

\noindent where
\be r= \left( \sum_{i=1}^{D}x_{i}^{2} \right)^{1/2}\ee

\noindent is the radial coordinate.

Let us consider a $D-$dimensional spherical $q-$Gaussian
wavefunction in the configuration space

\begin{equation}
\psi\left(\mathbf{r}\right)=C\left(1-\left(q-1\right) \beta
r^{2}\right)^{\frac{1}{2\left(q-1\right)}}
\label{eq:phi}\end{equation}

\noindent where $q$ and $\beta$ are positive parameters and $C$ is
an appropriate normalization constant. If $q<1$ the $q$-Gaussian
wavefunction (\ref{eq:phi}) remains finite for all ${\bf r}\in
{\mathbb R}^D$. On the other hand, when $q>1$ the $q$-Gaussian
vanishes at $r=1/\sqrt{(q-1)\beta}$ and is set to zero for
$r>1/\sqrt{(q-1)\beta}$ (see below a discussion on the physical
meaning of this cut-off). The space probability density $\rho({\bf
r}) = |\psi\left(\mathbf{r}\right)|^2$ associated with the
wavefunction (\ref{eq:phi}) maximizes Tsallis' power-law entropic
functional

\be S_q \, = \, \frac{1}{q-1} \, \left(1 - \int \rho^q \, d{\bf r}
\right)\ee

\noindent under the constraints given by normalization and the
expectation value of $r^2$ \cite{T09} (it can also be regarded as
a probability density maximizing R\'{e}nyi's functional under the same
constraints).

 Replacing (\ref{eq:phi}) into (\ref{eq:Schrodinger}) and
after some algebra, we find that the wavefunction is the ground
state of the potential

\begin{equation} V=\frac{\beta}{2} \left[\frac{-D+\beta
r^{2}\left(D\left(q-1\right)+3-2q\right)}{\left(1-\left(q-1\right)\beta
r^{2}\right)^{2}}\right]\label{eq:potential}\end{equation}

\noindent with eigenenergy equal to $0.$

In the case $q\le 1$ the potential function (\ref{eq:potential})
is finite for all ${\bf r} \in {\mathbb R}^D$. On the other hand,
when $q>1$ the potential function is singular when $r$ adopts the
particular value

\be r_w \, = \, \sqrt{\frac{1}{(q-1)\beta}}. \ee

\noindent Physically, this means that when $q>1$ the potential
function (\ref{eq:potential}) has an ``infinite wall" at $r=r_w$
and the quantum particle is confined within the region $r\le r_w$.
In this case, the $q$-Gaussian wavefunction (\ref{eq:phi})
vanishes at $r=r_w$ and must be set equal to zero when $r\ge r_w$.
This constitutes an example of the so-called Tsallis' cut-of
condition \cite{T09,PP93a}.

In the limit $q \rightarrow 1$ the $q$-Gaussian wavefunction
(\ref{eq:phi}) becomes a standard Gaussian, and the potential
function (\ref{eq:potential}) reduces to the $D$-dimensioanl
isotropic harmonic oscillator potential (notice that the origin of
the energy scale is shifted)

\be V(r) \, = \, -\frac{D\beta}{2} + \frac{1}{2} \beta^2 r^{2}.
\ee

\noindent The one dimensional instance of the potential
(\ref{eq:potential}) has been studied in \cite{RPGCP00}. This
potential exhibits the interesting feature of approximate shape
invariance (see \cite{RPGCP00} for details). This approximate
symmetry becomes exact in the limit $q\rightarrow 1$.

\section{q-Gaussian ground states in momentum space}

We now look for solutions solution of the Schrödinger equation
having the form of a $q-$Gaussian in momentum space

\begin{equation}
\tilde \psi\left(\mathbf{p}\right)=C\left(1-\left(q-1\right) \beta
p^{2}\right)^{\frac{1}{2\left(q-1\right)}}
\label{eq:phimo}\end{equation}

\noindent where

\be p^{2}=\sum_{i=1}^{D}p_{i}^{2}.\ee

\noindent As in the previous case of $q$-Gaussians in
configuration space, $q$ and $\beta$ are positive parameters and
$C$ is a normalization constant. We are going to consider
$q$-Gaussians in momentum space with $q<1$. Our aim is to
determine potential functions $V(r)$ having a ground state that,
in momentum space, has the form (\ref{eq:phimo}). In order to do
this, it will prove convenient not to work directly with the
Schrödinger equation in momentum space but, instead, to determine
first the Fourier transform $\psi\left(\mathbf{r}\right)$ of $
\tilde \psi\left(\mathbf{p}\right)$ and then to consider
Schrödinger's equation in configuration space.

The Fourier transform of the $q-$Gaussian wave function
(\ref{eq:phimo}) is

\begin{equation}
\psi_{\nu}\left(\mathbf{r}\right)=
\frac{2^{1-\nu}}{\Gamma\left(\nu\right)}
r^{\nu}K_{\nu}\left(r\right)\label{eq:psi}\end{equation}

\noindent where $K_{\nu}$ is the Bessel function of the second
kind and $r=\vert\mathbf{r}\vert.$  The parameter parameter $\nu$
is given by
\be \nu=-\frac{D}{2}-\frac{1}{2\left(q-1\right)}.\ee

\begin{thm*}
The function $\psi_{\nu}\left(r\right)$ is solution of the Schrödinger
equation associated with a potential

\begin{equation} \label{elpote}
V_{\nu}\left(r\right)= - \frac{1}{2} \,
\left(1+\frac{D}{2\left(\nu-1\right)}\right)
\frac{\psi_{\nu-1}\left(r\right)}{\psi_{\nu}\left(r\right)}.
\label{eq:Vnupotential}\end{equation}

As a special case, $ $when the parameter $\nu=d+\frac{1}{2}$ is
half integer, this potential is of the form

\begin{equation}
V_{d+\frac{1}{2}}\left(r\right)= - \frac{1}{2} \,
\left(1+\frac{D}{2d-1}\right)
\frac{p_{d-1}\left(r\right)}{p_{d}\left(r\right)}
\label{eq:Vdpotential}\end{equation}

\noindent where $p_{d}\left(r\right)$ is the Bessel polynomial of
degree $d.$
\end{thm*}
\begin{proof}
The derivation rule for the function $\psi\left(r\right)$ is

\be \frac{1}{r}\frac{\partial}{\partial r}
\psi_{\nu}\left(r\right)=
-\frac{1}{2\left(\nu-1\right)}\psi_{\nu-1}\left(r\right)\ee

\noindent so that the Laplace operator reads

\be \frac{1}{r^{D-1}}\frac{\partial}{\partial
r}\left(r^{D-1}\frac{\partial\psi_{\nu}}{\partial
r}\right)=\frac{1}{r^{D-1}}\left(\left(D-1\right)
r^{D-2}\frac{\partial\psi_{\nu}}{\partial
r}+r^{D-1}\frac{\partial^{2}\psi_{\nu}}{\partial r^{2}}\right)\ee

\noindent The first term is

\be
-\frac{\left(D-1\right)}{2\left(\nu-1\right)}\psi_{\nu-1}\left(r\right)
\ee

\noindent and the second term

\be \frac{\partial^{2}\psi_{\nu}}{\partial
r^{2}}=-\frac{1}{2\left(\nu-1\right)}\frac{\partial}{\partial
r}\left(r\psi_{\nu-1}\left(r\right)\right)=
-\frac{1}{2\left(\nu-1\right)}\psi_{\nu-1}
\left(r\right)+\frac{1}{4\left(\nu-1\right)
\left(\nu-2\right)}r^{2}\psi_{\nu-2}\left(r\right)\ee

\noindent so that the Laplace operator applied to $\psi_{\nu}$ is

\be \frac{1}{r^{D-1}}\frac{\partial}{\partial
r}\left(r^{D-1}\frac{\partial\psi_{\nu}}{\partial
r}\right)=\left(-\frac{D}{2\left(\nu-1\right)}
\psi_{\nu-1}\left(r\right)+\frac{1}{4\left(\nu-1\right)
\left(\nu-2\right)}r^{2}\psi_{\nu-2}\left(r\right)\right)\ee

\noindent Moreover, the function Bessel function $K_{\nu}$ obeys
the difference equation

\be rK_{\nu}\left(r\right)=rK_{\nu-2}\left(r\right)+
2\left(\nu-1\right)K_{\nu-1}\left(r\right)\ee

\noindent so that

\be r^{\nu}K_{\nu}\left(r\right)=
r^{2}r^{\nu-2}K_{\nu-2}\left(r\right)+
2\left(\nu-1\right)r^{\nu-1}K_{\nu-1}\left(r\right)\ee

\noindent and

\be \psi_{\nu}\left(r\right)= r^{2}\frac{1}{4\left(\nu-1\right)
\left(\nu-2\right)}\psi_{\nu-2}
\left(r\right)+\psi_{\nu-1}\left(r\right).\ee

\noindent We deduce

\be
\Delta\psi_{\nu}\left(r\right)=\frac{1}{r^{D-1}}\frac{\partial}{\partial
r}\left(r^{D-1}\frac{\partial\psi_{\nu}}{\partial
r}\right)=\left(-\frac{D}{2\left(\nu-1\right)}-1\right)
\psi_{\nu-1}\left(r\right)+\psi_{\nu}\left(r\right).\ee

\noindent Consequently,

\be -\frac{1}{2} \Delta\psi_{\nu}\left(r\right) + \left[
-\frac{1}{2}\left(1+\frac{D}{2\left(\nu-1\right)}\right)
\frac{\psi_{\nu-1}\left(r\right)}{\psi_{\nu}\left(r\right)}
\right] \, \psi_{\nu}\left(r\right) \, = \, -\frac{1}{2} \,
\psi_{\nu}\left(r\right), \ee

\noindent which means that $\psi_{\nu}\left(r\right) $ is an
eigenfunction of the potential $V_{\nu}(r) $ given by equation
(\ref{elpote}), with eigenvalue equal to $-\frac{1}{2}$.

\end{proof}

\section{Special cases and asymptotics}

\subsection{Asymptotics}

The asymptotics for large $r$ of the potential
(\ref{eq:Vnupotential}) can be computed using
\cite[9.7.2]{Abramowitz}

\be K_{\nu}\left(r\right)\sim\sqrt{\frac{\pi}{2r}}
e^{-r}\left(1+\frac{4\nu^{2}-1}{8r}+\dots\right)\ee

\noindent so that the asymptotics for the potential
(\ref{eq:Vnupotential}) reads

\be V_{\nu}\left(r\right)\sim - \frac{\left(2\left(\nu-1\right)+
D\right)}{r}\left(1+\frac{1}{2r}
\left(1-2\nu\right)+\dots\right).\ee

\noindent We see then that, for large values of $r$, the
asymptotic behavior of the potential $V_{\nu}\left(r\right)$ is
dominated by a Coulomb-like term.

\subsection{Special cases}
\begin{enumerate}
\item Coulomb potential: taking $\nu=\frac{1}{2}$
and remarking that
$\psi_{\frac{1}{2}}\left(r\right)=\exp\left(-r\right)$ and
$\psi_{-\frac{1}{2}}\left(r\right)=
-\frac{1}{r}\psi_{\frac{1}{2}}\left(r\right),$ we deduce that

\be -\frac{1}{2}\Delta\psi_{\frac{1}{2}}\left(r\right)
-\frac{D-1}{2 r}\psi_{\frac{1}{2}}\left(r\right)=
-\frac{1}{2}\psi_{\frac{1}{2}}\left(r\right)\ee

\noindent which is the Schrödinger equation associated with a
Coulomb potential. The associated probability density in
configurational space can be obtained as the squared modulus of
the inverse Fourier transform of the ground state wavefunction in
momentum space, \be \tilde{\psi}_{\frac{1}{2}}
\left(\mathbf{p}\right)\propto
\left(1+\vert\mathbf{p}\vert^{2}\right)^{-\frac{D+1}{2}}.\ee  The
momentum space representation of the eigenfunctions corresponding
to the $D$-dimensional Coulomb potential have been studied in
detail by Aquilanti, Cavalli and Coletti in \cite{aquilanti}.

The $q$-value characterizing the ground state of the
$-\frac{1}{r}$ potential is different from one. Indeed, it depends
on the value of the space dimension $D$,
\be q \, = \, \frac{D}{D+1}. \ee

\item In Table I we give a few potentials resulting from different half integer
values of $\nu$. These potentials are depicted in Fig.1.  If $\nu
= \frac{1}{2} + d$, with $d$ integer, then the entropic parameter
$q$ characterizing the $q$-Gaussian is given by
\be q \, = \, \frac{D+2d}{D+2d+1}. \ee

\noindent It is interesting that, for a given fixed value of $d$,
we have that $q \rightarrow 1$ when $D \rightarrow \infty$. That
is, when the space dimension tends to infinity the $q$-Gaussian
describing the ground state in momentum space approaches a
standard Gaussian.

\begin{table}
\begin{centering}
\begin{tabular}{|c|c|c|c|}
\hline $\nu$ & $\frac{3}{2}$ & $\frac{5}{2}$ &
$\frac{7}{2}$\tabularnewline \hline $V_{\nu}\left(r\right)$ &
$-\frac{1}{2}\left[ \frac{1+D}{1+r} \right]$ &
$-\frac{1}{2}\left[\frac{\left(3+D\right)\left(r+1\right)}{r^{2}+3r+3}\right]$
& $-\frac{1}{2}\left[\frac{\left(5+D\right)
\left(r^{2}+3r+3\right)}{r^{3}+6r^{2}+15r+15}\right]$\tabularnewline
\hline
\end{tabular}
\par\end{centering}
\caption{Values of the potential for different half-integer values
of the parameter $\nu$}
\end{table}
\end{enumerate}

\begin{figure}
\begin{center}
\includegraphics[scale=0.4]{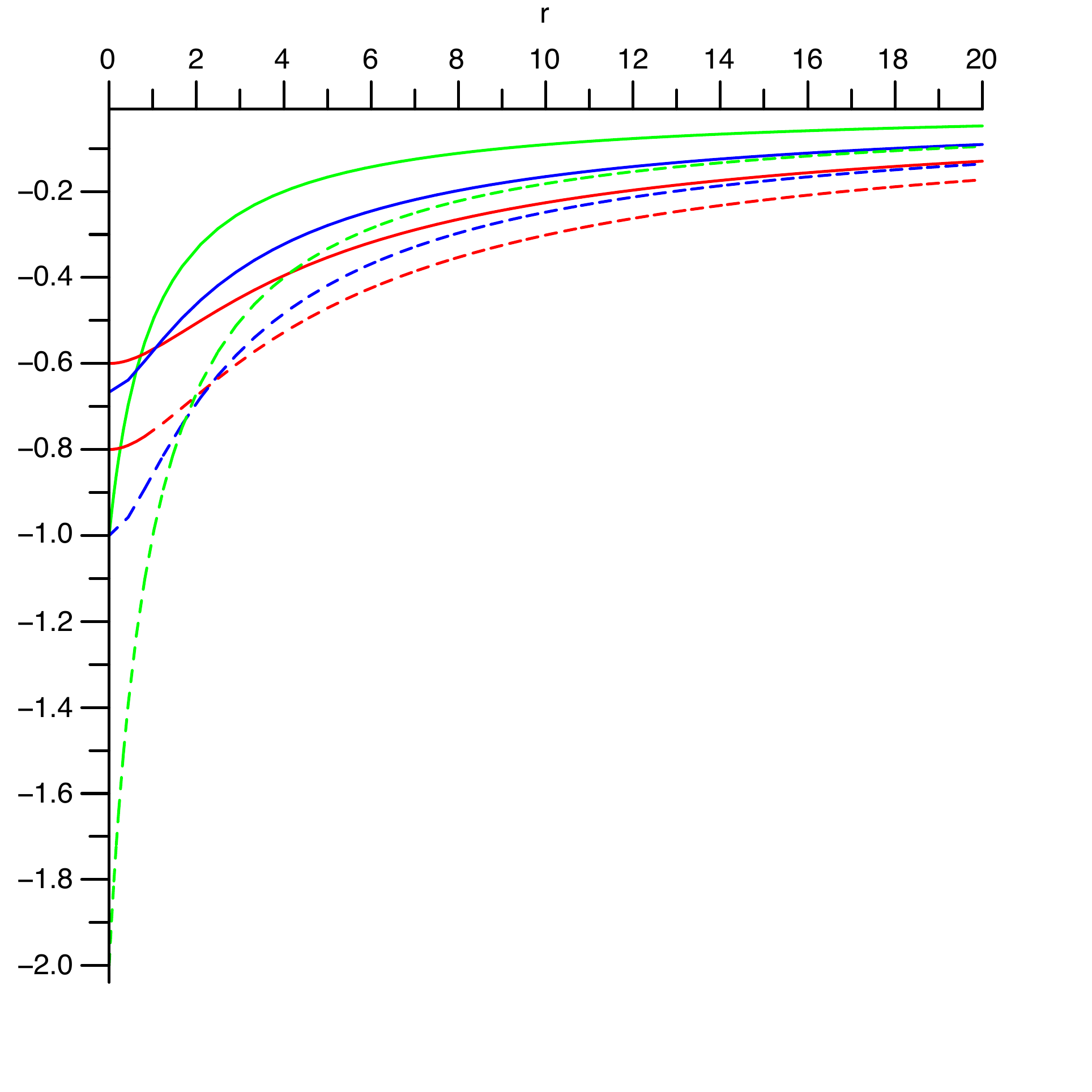}
\vskip -8mm \caption{The potential functions $V_{\nu}$ appearing
in Table I,  corresponding to $\nu$ equal to $\frac{3}{2}$,
$\frac{5}{3}$ and $\frac{7}{2}$
(top to bottom), for $D=1$ (solid line) and $D=3$
(dashed line).}
\end{center}
\label{Fig3}
\end{figure}

\section{Conclusions}

 We have determined the $D$-dimensional spherically symmetric
 potential functions having ground states of the $q$-gaussian
 form, either in configuration or in momentum space. In the case
 of $q$-gaussian ground states in configuration space we obtained a
 bi-parametric family of potentials admitting the $D$-dimensional
 isotropic harmonic oscillator as the particular case corresponding
 to the limit $q\rightarrow 1$. On the other hand, when
 considering ground states having the shape of a q-Gaussian in
 momentum space we obtained a family of potentials closely related
 to the $D$-dimensional Coulomb (or Hydrogen) potential $-\frac{1}{r}$.
 In point of fact, this family admits the standard
 ($D$-dimensional) Coulomb potential itself as a particular
 instance.
 Moreover, for large values of $r$, all the above mentioned
 potentials behave asymptotically as $-\frac{1}{r}$
  for all $0<q<1$.

  Within classical mechanics, it is already well known that
  there is a close relationship between the potential function
  $-\frac{1}{r}$ (describing Newtonian gravitation) and
  maximum $S_q$ distributions. The celebrated polytropic
  solutions of the Vlasov-Poisson equations, widely used in the study
  of self-gravitating astrophysical systems, have indeed the
  $S_q$-maxent form, and the associated velocity distributions
  are $q$-Gaussians. It is an intriguing fact that, as we have
  shown in the present work, there also exits a close connection
  between $q$-Gaussians and the $-\frac{1}{r}$ potential in
  quantum mechanics. \\

  \noindent
  {\bf Acknowledgments.} This work was partially supported by
  by the Project FQM-2445 of the Junta de Andaluc\'{\i}a
  and by the Grant FIS2008-2380 of the Ministerio de
  Innovaci\'on y  Ciencia, Spain.

\end{document}